\documentclass[conference]{IEEEtran}

\usepackage{cite}
\usepackage{amsthm, amsmath,amssymb,amsfonts,bbm}
\usepackage{algorithmic}
\usepackage[dvips]{graphicx}
\usepackage{subfigure}
\graphicspath{{graphics/}}
\usepackage{multirow}
\usepackage{filecontents, pgfplots}
\usepackage[algoruled]{algorithm2e}
\usepackage{microtype}
\usepackage{tikz}
\usepackage{textcomp}
\usepackage{xcolor}
\usepackage{verbatim}
\usepackage[binary-units=true]{siunitx}
\DeclareSIUnit{\belmilliwatt}{Bm}
\DeclareSIUnit{\dBm}{\deci\belmilliwatt}
\def\BibTeX{{\rm B\kern-.05em{\sc i\kern-.025em b}\kern-.08em
		T\kern-.1667em\lower.7ex\hbox{E}\kern-.125emX}}

\makeatletter
\newif\iftag@here

\newcommand*{\taghere}[1][0pt]
{\ifmeasuring@\else
	\global\tag@heretrue
	\tikz[remember picture,overlay]{\coordinate (taghere) at (0pt,#1);}%
	\fi}

\def\place@tag{%
	\iftagsleft@
	\kern-\tagshift@
	\iftag@here
	\global\tag@herefalse
	\tikz[remember picture,overlay]%
	{\path (taghere) -| node[anchor=base]{\rlap{\boxz@}} (0pt,0pt);}%
	\else
	\if1\shift@tag\row@\relax
	\rlap{\vbox{%
			\normalbaselines
			\boxz@
			\vbox to\lineht@{}%
			\raise@tag
	}}%
	\else
	\rlap{\boxz@}%
	\fi
	\kern\displaywidth@
	\fi
	\else
	\kern-\tagshift@
	\iftag@here
	\global\tag@herefalse
	\tikz[remember picture,overlay]%
	{\path  (taghere) -|  node[anchor=base]{\llap{\boxz@}} (0pt,0pt);}%
	\else
	\if1\shift@tag\row@\relax
	\llap{\vtop{%
			\raise@tag
			\normalbaselines
			\setbox\@ne\null
			\dp\@ne\lineht@
			\box\@ne
			\boxz@
	}}%
	\else \llap{\boxz@}%
	\fi
	\fi
	\fi
}
\makeatother

\DeclareMathOperator*{\argmax}{arg\,max}
\DeclareMathOperator*{\maximize}{maximize}
\DeclareMathOperator*{\argmin}{arg\,min}

\DeclareMathOperator*{\subjectto}{subject\,to}

\usepackage[acronym,nomain]{glossaries}
\makeglossaries
\newacronym{swipt}{SWIPT}{simultaneous wireless information and power transfer}
\newacronym{wpt}{WPT}{wireless power transfer}
\newacronym{wit}{WIT}{wireless information transfer}
\newacronym{iot}{IoT}{Internet-of-Things}
\newacronym{awgn}{AWGN}{additive white Gaussian noise}
\newacronym{tx}{TX}{transmitter}
\newacronym{ir}{IR}{information receiver}
\newacronym{eh}{EH}{energy harvester}
\newacronym{ap}{AP}{average power}
\newacronym{pp}{PP}{peak power}

\newacronym{siso}{SISO}{single-input single-output}
\newacronym{mimo}{MIMO}{multiple-input multiple-output}
\newacronym{miso}{MISO}{multiple-input single-output}
\newacronym{simo}{SIMO}{single-input multiple-output}

\newacronym{rf}{RF}{radio frequency}
\newacronym{dc}{DC}{direct current}
\newacronym{ac}{AC}{alternative current}
\newacronym{papr}{PAPR}{peak-to-average power ratio}
\newacronym{lp}{LPF}{low-pass filter}
\newacronym{mc}{MC}{matching circuit}
\newacronym{mrt}{MRT}{maximum ratio transmission}

\newacronym{rv}{RV}{random variable}
\newacronym{iid}{i.i.d.}{independent and identically distributed}
\newacronym{pdf}{pdf}{probability density function}

\newacronym{dnn}{DNN}{dense neural networks}
\glsunset{dnn}
\newacronym{mdp}{MDP}{Markov decision process}

\newacronym{sca}{SCA}{successive convex approximation}
\newacronym{sdr}{SDR}{semi-definite relaxation}

\newacronym{spr}{LP}{low power}
\newacronym{mpr}{MP}{medium power}
\newacronym{lpr}{HP}{high power}

\DeclareMathOperator{\rank}{rank}
\newcommand{\norm}[1]{\left\lVert#1\right\rVert_2}
\newcommand{\Tr}[1]{\text{Tr}\{#1\} }

\begin{document}

	\newtheorem{proposition}{Proposition}	
	\newtheorem{lemma}{Lemma}	
	\newtheorem{theorem}{Theorem}	
	\newtheorem{corollary}{Corollary}
	\newtheorem{assumption}{Assumption}	
	\newtheorem{remark}{Remark}	
	
	\title{Optimal Transmit Strategy for MIMO WPT Systems With Non-linear Energy Harvesting}
	
	\author{\IEEEauthorblockN{Nikita Shanin, Laura Cottatellucci, and Robert Schober}
	\IEEEauthorblockA{\textit{Friedrich-Alexander-Universit\"{a}t Erlangen-N\"{u}rnberg (FAU), Germany} }	}

	\IEEEspecialpapernotice{(Invited Paper)}
	
	\maketitle
	
\begin{abstract}
In this paper, we study multiple-input multiple-output (MIMO) wireless power transfer (WPT) systems, where the energy harvester (EH) node is equipped with multiple non-linear rectennas.
We characterize the optimal transmit strategy by the optimal distribution of the transmit symbol vector that maximizes the average harvested power at the EH subject to a constraint on the power budget of the transmitter.
We show that the optimal transmit strategy employs scalar unit-norm input symbols with arbitrary phase and two beamforming vectors, which are determined as solutions of a non-convex optimization problem.
To solve this problem, we propose an iterative algorithm based on a two-dimensional grid search, semi-definite relaxation, and successive convex approximation.
Our simulation results reveal that the proposed MIMO WPT design significantly outperforms two baseline schemes based on a linear EH model and a single beamforming vector, respectively.
Finally, we show that the average harvested power grows linearly with the number of rectennas at the EH node and saturates for a large number of TX antennas.
\end{abstract}
\setcounter{footnote}{0}
 
	\section{Introduction}
	\label{Section:Introduction}	
	Recently, \gls*{wpt} has attracted significant attention as a promising technology for next-generation communication networks since it enables remote recharging of the batteries of low-power \gls*{iot} devices \cite{Grover2010, Zhang2013, Boshkovska2015, Kim2020, Boshkovska2017a, Ma2019, Clerckx2018, Shen2020, Morsi2019, Shanin2020, Shanin2021a}.
First, in \cite{Grover2010}, the authors studied \gls*{siso} WPT systems and showed that broadcasting a single sinusoidal signal is optimal for the maximization of the power delivered to the \gls*{eh}. 
Then, in \cite{Zhang2013}, the authors extended the results in \cite{Grover2010} to \gls*{mimo} WPT systems, where the received power at the EH is maximized if a scalar input symbol and \emph{energy beamforming} are employed at the \gls*{tx}.
However, in \cite{Boshkovska2015} and \cite{Kim2020}, the authors showed that practical EH circuits are non-linear and, hence, the solutions in \cite{Grover2010} and \cite{Zhang2013} are not necessarily optimal for the maximization of the power harvested by the EH.
Therefore, an accurate modeling of the EH is essential for the design of practical \gls*{wpt} systems \cite{Boshkovska2015, Kim2020, Boshkovska2017a, Ma2019, Clerckx2018, Shen2020, Morsi2019, Shanin2020, Shanin2021a}.

Practical EHs typically employ a rectenna, i.e., an antenna connected to a rectifier, which comprises a non-linear circuit featuring a diode.
In \cite{Boshkovska2015} and \cite{Kim2020}, the authors showed that while, in the low input power regime, the rectenna non-linearity is determined by the non-linear forward-bias current-voltage characteristic of the rectifying diode, the rectennas also typically suffer from saturation in the high input power regime due to the breakdown effect of the diode.
Furthermore, the authors in \cite{Boshkovska2015} proposed a rectenna model, where, for signals with a fixed known waveform, the harvested power is specified as a parameterized sigmoidal function of the power of the received signal. 
This model was exploited for the design of MIMO WPT systems in, e.g., \cite{Boshkovska2017a, Ma2019}, where Gaussian signals were assumed.
Next, the authors in \cite{Clerckx2018} derived a non-linear EH model based on the Taylor series expansion of the current flow through the rectifying diode.
Exploiting this model, the authors studied a \gls*{miso} WPT system with a single-antenna EH and showed that the harvested power is maximized with energy beamforming \cite{Zhang2013}, which reduces to \gls*{mrt} in this case.
Furthermore, in \cite{Shen2020}, using the model in \cite{Clerckx2018}, the authors considered MIMO WPT systems, whose EH nodes were equipped with multiple rectennas, and proposed an iterative algorithm to determine the transmit beamforming vector that maximizes the harvested power at the EH node.

Although the results in \cite{Boshkovska2015, Kim2020, Boshkovska2017a, Ma2019, Clerckx2018, Shen2020} provide important insights for the design of MIMO WPT systems, their applicability is limited since, for the model in \cite{Boshkovska2015}, signals with a fixed known waveform are assumed, and for the model in \cite{Clerckx2018}, the saturation of the harvested power is neglected.
Therefore, the authors in \cite{Morsi2019} derived a realistic EH model that accurately captures the rectenna non-linearity in both the low and high input power regimes.
Furthermore, in \cite{Morsi2019}, the authors showed that for \gls*{siso} WPT systems, it is optimal to adopt ON-OFF signaling at the TX.
The ON symbol and its probability were chosen to maximize the average harvested power without saturating the EH while satisfying the average power constraint of the TX.
The optimality of ON-OFF signaling for SISO WPT systems was confirmed in \cite{Shanin2020}, where a learning-based approach was employed to model non-linear rectenna circuits equipped with a single and multiple diodes, respectively.
Finally, in \cite{Shanin2021a}, exploiting the EH model in \cite{Morsi2019}, the authors studied the harvested power region of two-user \gls*{miso} WPT systems.
However, to the best of the authors' knowledge, the problem of optimizing the transmit strategy for MIMO WPT systems, for the case where the EH node is equipped with multiple rectennas exhibiting non-linear behavior in both the low and high input power regimes, has not been solved in the literature, yet.

In this paper, we determine the optimal transmit strategy for the maximization of the average harvested power at the EH node for MIMO WPT systems.
To this end, we consider a multi-rectenna EH model, where each antenna is connected to a dedicated rectifier and adopt the rectenna model derived in \cite{Morsi2019} to take into account the EH non-linearity.
Then, we formulate an optimization problem for the maximization of the average harvested power at the EH under a constraint on the power budget of the TX.
We show that the optimal transmit strategy employs scalar unit-norm input symbols with arbitrary phase and two beamforming vectors.
To determine these beamforming vectors, we propose a low-complexity iterative algorithm based on a two-dimensional grid search, \gls*{sdr}, and \gls*{sca} \cite{Grippo2000, Xu2019, Sun2017}.
Our simulation results reveal that the proposed MIMO WPT design outperforms two baseline schemes, which are based on a linear EH model and a single beamforming vector, respectively.
Furthermore, we observe that the average harvested power grows practically linearly with the number of rectennas equipped at the EH and saturates for a large number of TX antennas.

	The remainder of this paper is organized as follows. 
	In Section II, we introduce the system model and the adopted EH model.
	In Section III, we formulate and solve the optimization problem for the maximization of the harvested power.
	In Section IV, we provide numerical results to evaluate the performance of the proposed design.
	Finally, in Section V, we draw some conclusions.
	
	\emph{Notation:} Bold upper case letters $\boldsymbol{X}$ represent matrices and ${X}_{i,j}$ denotes the element of $\boldsymbol{X}$ in row $i$ and column $j$. 
	Bold lower case letters $\boldsymbol{x}$ stand for vectors and ${x}_{i}$ is the $i^\text{th}$ element of $\boldsymbol{x}$.
	$\boldsymbol{X}^H$, $\Tr{\boldsymbol{X}}$, and $\rank \{\boldsymbol{X}\}$ denote the Hermitian, trace, and rank of matrix $\boldsymbol{X}$, respectively.
	The expectation with respect to random variable $x$ is denoted by $\mathbb{E}_x\{\cdot\}$. 
	The real part of a complex number is denoted by $\Re\{ \cdot \}$.
	$\boldsymbol{x}^\top$ and $\norm{\boldsymbol{x}}$ represent the transpose and L2-norm of $\boldsymbol{x}$, respectively.
	The imaginary unit is denoted by $j$.
	The sets of real and complex numbers are denoted by $\mathbb{R}$ and $\mathbb{C}$, respectively.
	$\boldsymbol{1}_K$ and $\boldsymbol{0}_K$ represent column vectors comprising $K$ elements, where all elements are equal to $1$ and $0$, respectively.
	The Dirac delta function is denoted by $\delta(x)$.

	\section{System Model}
	\label{Section:SystemModelPreliminaries}
	In this section, we present the MIMO WPT system model and discuss the adopted multi-antenna EH model.
		\subsection{System Model}
		\label{Section:SystemModel}
		\begin{figure}[!t]
	\centering
	\includegraphics[width=0.4\textwidth, draft = false]{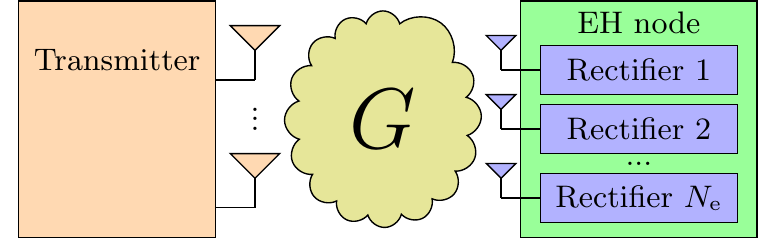}
	\caption{MIMO WPT system comprising multi-antenna TX and EH nodes. Here, $\boldsymbol{G}$ denotes the channel between the TX and EH.}
	\label{Fig:SystemModel}
	\vspace*{-10pt}
\end{figure}
We consider a narrow-band MIMO \gls*{wpt} system comprising a \gls*{tx} and an \gls*{eh} equipped with $N_\text{t} \geq 1$ and $N_\text{e} \geq 1$ antennas, respectively, see Fig.~\ref{Fig:SystemModel}.
The TX broadcasts a pulse-modulated RF signal, whose equivalent complex baseband (ECB) representation is modeled as $\boldsymbol{x}(t) = \sum_n \boldsymbol{x}[n] \psi(t-nT)$, where $\boldsymbol{x}[n] \in \mathbb{C}^{N_\text{t}}$ is the transmitted vector in time slot $n$, $\psi(t)$ is the rectangular transmit pulse shape, and $T$ is the symbol duration.
Transmit vectors $\boldsymbol{x}[n]$ are mutually independent realizations of a random vector $\boldsymbol{x}$, whose \gls*{pdf} is denoted by $p_{\boldsymbol{x}}(\boldsymbol{x})$.

The ECB channel between the TX and antenna $p$ of the EH node is characterized by row-vector $\boldsymbol{g}_p \in \mathbb{C}^{1 \times N_\text{t}}$, $p \in \{1,2,...,N_\text{e}\}$.
Thus, the RF signal received at antenna $p$ of the EH node is given by ${z^\text{RF}_{p}}(t) = \sqrt{2} \Re\{ \boldsymbol{g}_p \boldsymbol{x}(t) \exp(j 2 \pi f_c t) \} $, where $f_c$ denotes the carrier frequency.
The noise received at the EH node is ignored since its contribution to the harvested power is negligible.

		\subsection{Energy Harvester Model}
		\label{Section:EhModel}
		In this paper, we assume that the EH node is equipped with $N_\text{e}$ rectennas, i.e., each antenna is connected to a dedicated rectifier using a matching circuit, see Fig.~\ref{Fig:SystemModel}, \cite{Morsi2019, Shanin2020}.
In order to maximize the power delivered to the rectifier, the matching circuit is typically well-tuned to the carrier frequency $f_c$ and is designed to match the input impedance of the non-linear rectifier with the output impedance of the antenna \cite{Morsi2019}.
The non-linear rectifier, which comprises a rectifying diode and a low-pass filter, converts the RF signal ${z^\text{RF}_{p}}(t)$ received by rectenna $p$ to a \gls*{dc} signal at the load resistor $R_\text{L}$.

In this paper, we adopt the rectenna model derived in \cite{Morsi2019}. 
Thus, the power harvested by rectenna $p, p\in\{1,2,...,N_\text{e}\},$ as a function of magnitude $|z_p|$ of the received ECB signal $z_p = \boldsymbol{g}_p \boldsymbol{x}$ is modeled as\footnotemark
\begin{equation}
	\phi(|z_p|^2) = \min \{\varphi(|z_p|^2), \varphi(A_s^2)\},\label{Eqn:RaniaModel}
\end{equation}
\noindent\hspace*{-1pt}where $\varphi(|z_p|^2) = \Big[\frac{1}{a} W_0 \Big(a\exp(a) I_0 \big(B\sqrt{2|{z_p}|^2}\big) \Big)-1 \Big]^2 I_s^2 R_L$, $a = \frac{I_s(R_L + R_s)}{\mu V_\text{T}}$, $B = [\mu V_\text{T} \sqrt{\Re \{1 / Z_a^* \} }]^{-1}$, and $W_0(\cdot)$ and $I_0(\cdot)$ are the principal branch of the Lambert-W function and the zeroth order modified Bessel function of the first kind, respectively.
\footnotetext{In this paper, we assume that all rectennas of the EH node are memoryless and have identical parameters. Therefore, we omit the time slot index $n$ and model all rectennas by the same function in (\ref{Eqn:RaniaModel}).}
Here, $Z_a^*$, $V_\text{T}$, $I_s$, $R_s$, and $\mu \in [1,2]$ are parameters of the rectenna circuit, namely, the complex-conjugate of the input impedance of the rectifier circuit, the thermal voltage, the reverse bias saturation current, the series antenna resistance, and the ideality factor of the diode, respectively. 
These parameters are determined by the circuit elements and are independent of the received signal.
Since, for large input power levels, rectenna circuits are driven into saturation \cite{Boshkovska2015, Morsi2019, Shanin2020}, the function in (\ref{Eqn:RaniaModel}) is bounded, i.e., $\phi(|z_p|^2) \leq \phi(A_s^2) , \; \forall z_p \in \mathbb{C}$, where $A_s$ is the minimum input signal magnitude level at which the output power starts to saturate.
Finally, we define the total harvested power as the sum of the powers harvested by the rectennas of the EH node, i.e., $\psi(\boldsymbol{x}) = \sum_{p = 1}^{N_\text{e}} \phi(|\boldsymbol{g}_p \boldsymbol{x}|^2)$.

		\section{Problem Formulation and Solution}
		\label{Section:ProblemFormulation}
		In this section, we formulate an optimization problem for the maximization of the average harvested power of the considered MIMO WPT system.
		Then, in order to solve this problem, we formulate and solve a related auxiliary optimization problem, where we maximize the expected value of a one-dimensional function of a random variable under a constraint on its mean value.
		Finally, we solve the original optimization problem and provide the optimal transmit strategy.
		\subsection{Problem Formulation}
		In this paper, we characterize the optimal transmit strategy by the pdf $p^*_{\boldsymbol{x}}(\boldsymbol{x})$ of transmit symbol vector $\boldsymbol{x}$ that maximizes the average total harvested power at the EH node.
Thus, we formulate the following optimization problem:
\begin{subequations}
	\begin{align}
	\maximize_{ {p}_{\boldsymbol{x}} } \quad &\overline{\Phi}(\boldsymbol{x}; p_{\boldsymbol{x}}) 
	\label{Eqn:WPT_GeneralProblem_Obj} \\
	\subjectto \quad & \int_{\boldsymbol{x}} \norm{\boldsymbol{x}}^2 p_{\boldsymbol{x}}(\boldsymbol{x}) d \boldsymbol{x} \leq P_x, \label{Eqn:WPT_GeneralProblem_C1}\\
	&\int_{\boldsymbol{x}} p_{\boldsymbol{x}}(\boldsymbol{x}) d \boldsymbol{x} = 1, \label{Eqn:WPT_GeneralProblem_C2}
	\end{align}
	\label{Eqn:WPT_GeneralProblem}
\end{subequations}
\noindent\hspace*{-3.5pt}where the average harvested power at the EH node is defined as
$\overline{\Phi}(\boldsymbol{x}; p_{\boldsymbol{x}}) = \mathbb{E}_{\boldsymbol{x}} \{\psi(\boldsymbol{x})\}$.
We impose constraints (\ref{Eqn:WPT_GeneralProblem_C1}) and (\ref{Eqn:WPT_GeneralProblem_C2}) to limit the transmit power budget at the TX and ensure that $p_{\boldsymbol{x}}(\boldsymbol{x})$ is a valid pdf, respectively.

\begin{remark}
	\label{Remark:NotUniqueness}
	Optimization problem (\ref{Eqn:WPT_GeneralProblem}) may have an infinite number of solutions. 
	In particular, since $\| \boldsymbol{x} \|_2$, the harvested power in (\ref{Eqn:RaniaModel}) and, hence, the average harvested power $\overline{\Phi}(\cdot)$ are invariant under phase rotation of the transmit symbol vector $\boldsymbol{x}$, given an optimal pdf $p^*_{\boldsymbol{x}}(\boldsymbol{x})$ and a random phase $\phi_x \in [-\pi, \pi)$ with an arbitrary pdf $p_{\phi_x}(\phi_x)$, the random vector $\boldsymbol{\tilde{x}}=\exp(j\phi_x)\boldsymbol{x}$ is still a solution of (\ref{Eqn:WPT_GeneralProblem}) \cite{Shanin2020}.
	Furthermore, we note that if affordable by the power budget $P_x$, there may be an infinite number of pdfs that drive all the rectifiers of the EH into saturation, i.e., yield $\overline{\Phi}(\boldsymbol{x}; p_{\boldsymbol{x}}) = N_\text{\upshape e} \phi(A_s^2)$ and satisfy constraints (\ref{Eqn:WPT_GeneralProblem_C1}) and (\ref{Eqn:WPT_GeneralProblem_C2}).
	In the following, we determine one pdf $p_{\boldsymbol{x}}^*$ that solves (\ref{Eqn:WPT_GeneralProblem}).	
\end{remark}

In order to find the optimal solution of (\ref{Eqn:WPT_GeneralProblem}), it is convenient to solve first a related auxiliary optimization problem.
To this end, in the next subsection, we consider the maximization of the expectation of a non-decreasing function $f(\nu)$ of a scalar random variable $\nu$ under a constraint on the mean value of $\nu$.

		\subsection{Auxiliary Optimization Problem}
		Let us consider the following auxiliary optimization problem:
\begin{equation}
	\maximize_{ {p}_{\nu} } \; \mathbb{E}_\nu \{f(\nu)\} \quad \subjectto \; \mathbb{E}_\nu \{\nu\} \leq \overline{\nu},
	\label{Eqn:GeneralOptimizationProblem}
\end{equation}
\noindent where we optimize the pdf $p_\nu(\nu)$ of $\nu$ for the maximization of the expectation of $f(\nu)$ under a constraint $\overline{\nu}$ on the mean value of $\nu$.
In order to solve (\ref{Eqn:GeneralOptimizationProblem}), let us first define the slope of the straight line connecting points $(\nu_1, f(\nu_1))$ and $(\nu_2, f(\nu_2))$, where $\nu_2 > \nu_1$, as follows:
\begin{equation}
	s(\nu_1,\nu_2) = \frac{f(\nu_2) - f(\nu_1)}{\nu_2 - \nu_1}.
	\label{Eqn:SlopeFunctionF}
\end{equation}
We note that if $f(\nu)$ is convex (concave), the solution of (\ref{Eqn:GeneralOptimizationProblem}) is determined by the Edmundson-Madansky (Jensen's) inequality, see, e.g., \cite{Dokov2002}. 
However, since we intend to apply this result to function $\psi(\cdot)$, which is neither convex nor concave, in the following lemma, we extend the results in \cite{Dokov2002} to arbitrary non-decreasing functions $f(\nu)$ and solve (\ref{Eqn:GeneralOptimizationProblem}).
\begin{lemma}
	The solution\footnotemark\hspace*{3pt}of optimization problem (\ref{Eqn:GeneralOptimizationProblem}) is a discrete pdf given by $p_\nu^*(\nu) = \beta \delta(\nu - \nu_1^*) + (1-\beta) \delta(\nu - \nu_2^*)$, where
	$\nu_1^* = \argmin_{\nu_1 \leq \overline{\nu} } c(\nu_1)$, $c(\nu_1) = \max_{\nu_2 \geq \overline{\nu}} s(\nu_1,\nu_2)$, $\nu_2^* = \argmax_{\nu_2 \geq \overline{\nu}} s(\nu_1^*,\nu_2)$, and $\beta = \frac{\nu_2^* - \overline{\nu}}{\nu_2^* - \nu_1^*}$.
	\label{Theorem:Corollary2}	
\end{lemma}
\begin{proof}
	Please refer to \cite[Appendix A]{Shanin2021a}.
\end{proof}
\footnotetext{We note that, similar to (\ref{Eqn:WPT_GeneralProblem}), problem (\ref{Eqn:GeneralOptimizationProblem}) may have an infinite number of solutions, i.e., for a monotonic non-decreasing function $f(\cdot)$, there may exist multiple pdfs $p_{\nu}(\nu)$ that yield equal values $\mathbb{E}_\nu \{f(\nu)\}$ and satisfy the constraint in (\ref{Eqn:GeneralOptimizationProblem}).
In Lemma~\ref{Theorem:Corollary2}, we obtain one solution of (\ref{Eqn:GeneralOptimizationProblem}).}

In the following, we exploit the result in Lemma~\ref{Theorem:Corollary2} for solving the problem in (\ref{Eqn:WPT_GeneralProblem}).

		\subsection{Solution of Problem (\ref{Eqn:WPT_GeneralProblem})}
		\label{Section:MIMOSystem}
		In the following, we consider the MIMO WPT system in Fig.~\ref{Fig:SystemModel}.
In Proposition~\ref{Theorem:Proposition3}, we characterize the optimal pdf $p^*_{\boldsymbol{x}}(\boldsymbol{x})$ that solves (\ref{Eqn:WPT_GeneralProblem}).
\begin{proposition}
	For the considered MIMO WPT system, function $\overline{\Phi} (\cdot)$ is maximized for transmit symbol vectors $\boldsymbol{x} = \boldsymbol{w} s$, where $s = \exp(j \phi_s)$ is a unit norm symbol with an arbitrary phase $\phi_s$.
	Here, $\boldsymbol{w}$ is a discrete random beamforming vector, whose pdf is given by $p^*_{\boldsymbol{w}}(\boldsymbol{w}) = \beta \delta(\boldsymbol{w} - \boldsymbol{w}^*_1) + (1-\beta) \delta(\boldsymbol{w} - \boldsymbol{w}^*_2)$. 
	The beamforming vectors $\boldsymbol{w}^*_n, \; n\in\{1,2\}$, are given by
	\begin{align}
	\boldsymbol{w}^*_n &\in \{ \boldsymbol{w} : \psi(\boldsymbol{w}) = \Phi(\nu_n^*) \},
	\label{Eqn:MimoPropositionBeamformerProblem} \\
	\Phi(\nu) &= \max_{ \{ \boldsymbol{w} \, | \, \boldsymbol{w} \in \mathbb{C}^{N_\text{\upshape t}}, \norm{\boldsymbol{w}}^2 = \nu \}  } \psi(\boldsymbol{w}),
	\label{Eqn:MimoPropositionFunction}
	\end{align} 
	where $\nu^*_1$ and $\nu^*_2$ are defined as in Lemma~\ref{Theorem:Corollary2} with $\overline{\nu} = P_x$ and $s(\nu_1, \nu_2) = \big(\Phi(\nu_2) - \Phi(\nu_1)\big) / \big(\nu_2 - \nu_1\big)$.
	\label{Theorem:Proposition3}
\end{proposition}
\begin{proof}
	Please refer to Appendix~\ref{Appendix:Prop3Proof}.
\end{proof}

Proposition~\ref{Theorem:Proposition3} reveals that the optimal transmit vector $\boldsymbol{x}$ is discrete and is characterized by scalar unit-norm symbols with an arbitrary phase\footnotemark\hspace*{0pt} and two beamforming vectors, $\boldsymbol{w}^*_1$ and $\boldsymbol{w}^*_2$.
We note that in order to determine these beamforming vectors, the function $\Phi(\cdot)$ and the values $\nu_1^*$, $\nu_2^*$ as solutions of (\ref{Eqn:MimoPropositionFunction}) and the non-convex min-max optimization problems in Lemma~\ref{Theorem:Corollary2}, respectively, are required.
In the following, first, exploiting \gls*{sdr} and \gls*{sca}, we develop a low-complexity algorithm to determine a suboptimal solution for (\ref{Eqn:MimoPropositionFunction}).
Then, due to the low dimensionality of the min-max problem in Lemma~\ref{Theorem:Corollary2}, we obtain the optimal values $\nu_1^*$, $\nu_2^*$ and the corresponding vectors $\boldsymbol{w}_1^*$, $\boldsymbol{w}_2^*$ via a two-dimensional grid search \cite{Coope2001}.
\footnotetext{We note that the phase $\phi_s$ of scalar symbol $s$ can be chosen arbitrarily in each time slot $n$. This degree of freedom can be further exploited, for example, for information transmission \cite{Shanin2020}.}

			\subsubsection{Solution of (\ref{Eqn:MimoPropositionFunction})}
			\label{Section:SuboptimalSolution}
			In the following, we exploit \gls*{sdr} and \gls*{sca} to develop an iterative low-complexity algorithm and determine a suboptimal solution of (\ref{Eqn:MimoPropositionFunction}).
To this end, we first define matrix $\boldsymbol{W} = \boldsymbol{w} \boldsymbol{w}^H$ and reformulate problem (\ref{Eqn:MimoPropositionFunction}) equivalently as follows:
\begin{subequations}
	\begin{align}
	\maximize_{\boldsymbol{W} \in \mathcal{S}_{+}} \quad & \Psi(\boldsymbol{W})
	\label{Eqn:MimoSuboptimalFunctionRef_Obj} \\
	\subjectto \quad &\Tr{\boldsymbol{W}} \leq \nu, \label{Eqn:MimoSuboptimalFunctionRef_C1} \\
	&\rank \{ \boldsymbol{W} \} = 1, \label{Eqn:MimoSuboptimalFunctionRef_C2}
	\end{align}	
	\label{Eqn:MimoSuboptimalFunctionRef}
\end{subequations}
\noindent\hspace*{-4pt}where $\Psi(\boldsymbol{W}) = \sum_{p = 1}^{N_\text{e}} \phi(\boldsymbol{g}_p \boldsymbol{W} \boldsymbol{g}_p^H)$ and $\mathcal{S}_{+}$ denotes the set of positive semidefinite matrices.
Since the objective function in (\ref{Eqn:MimoPropositionFunction}) is monotonic non-decreasing in $|\boldsymbol{w}| = \big[|w_1| \, |w_2| \, ... \, |w_{N_\text{t}}|\big]^\top$\hspace*{-3pt}, we equivalently replace the equality constraint in (\ref{Eqn:MimoPropositionFunction}) by inequality constraint (\ref{Eqn:MimoSuboptimalFunctionRef_C1}).

We note that optimization problem (\ref{Eqn:MimoSuboptimalFunctionRef}) is non-convex due to the non-convexity of objective function (\ref{Eqn:MimoSuboptimalFunctionRef_Obj}) and constraint (\ref{Eqn:MimoSuboptimalFunctionRef_C2}). 
Therefore, in order to obtain a suboptimal solution of (\ref{Eqn:MimoSuboptimalFunctionRef}), we first drop constraint (\ref{Eqn:MimoSuboptimalFunctionRef_C2}).
Then, we define sets $\mathcal{W}_k, k\in\{0,1,...,N_\text{e}\}$, where $\forall \boldsymbol{W} \in \mathcal{W}_k$, exactly $k$ rectifiers are driven into saturation.
We note that $\mathcal{W}_0 \cup \mathcal{W}_1 \cup ... \cup \mathcal{W}_{N_\text{e}} = \mathcal{S}_{+}$.
Furthermore, rectifier $p$ of the EH node is driven into saturation if and only if $\boldsymbol{g}_{p} \boldsymbol{W} \boldsymbol{g}_{p}^H \geq A_s^2$.
Hence, set $\mathcal{W}_k, k\in\{0,1,...,N_\text{e}\},$ consists of $T_k = \frac{N_\text{e}!}{k!(N_\text{e}-k)!}$ convex subsets, where $k!$ denotes the factorial of $k$, i.e., $\mathcal{W}_k = \mathcal{W}_k^1 \cup \mathcal{W}_k^2 \cup ... \cup \mathcal{W}^{T_k}_k$.
Each convex subset $\mathcal{W}^t_k, t\in\{1,2,...,T_k\},$ specifies a unique combination of $k$ rectennas, which are driven into saturation $\forall \, \boldsymbol{W} \in \mathcal{W}^t_k$.
We note that the objective function in (\ref{Eqn:MimoSuboptimalFunctionRef}) is convex on each of these subsets and, thus, applying \gls*{sca} for solving (\ref{Eqn:MimoSuboptimalFunctionRef}) on $\boldsymbol{W} \in \mathcal{W}^t_k$ is promising \cite{Sun2017, Lanckriet2009}.
Hence, the solution of (\ref{Eqn:MimoSuboptimalFunctionRef}) can be obtained by exploring the subsets $\mathcal{W}^t_k, t\in\{1,2,...,T_k\}, k\in\{0,1,...,N_\text{e}\}$ and solving the resulting optimization problem \cite{Shanin2021a}.
However, since the computational complexity of this exploration grows with $N_\text{e}!$, in the following, we propose a suboptimal solution of (\ref{Eqn:MimoSuboptimalFunctionRef}).

For a given transmit power limit $\nu$, we determine a set of rectennas $\mathcal{W}^*$, which will be driven into saturation.
To this end, we sort the channel gain vectors $\boldsymbol{g}_p$ in descending order of their norms as follows $\norm{\boldsymbol{g}_{p^1}} \geq \norm{\boldsymbol{g}_{p^2}} \geq ... \geq \norm{\boldsymbol{g}_{p^{N_\text{e}}}}$, where $p^k \in \{1,2,..,N_\text{e}\}$ and $k = 1,2,..,N_\text{e}$.
Then, we check if it is possible to drive the $k$ rectifiers with the best channel conditions, i.e., rectifiers $p^1, p^2, .., p^k$, into saturation.
To this end, we solve the following convex optimization problem:
\begin{subequations}
	\begin{align}
	\maximize_{\boldsymbol{W} \in \mathcal{S}_{+}} \quad  &1   \\
	\subjectto \quad & (-1)^{u_n} (\boldsymbol{g}_{p^n} 
	\boldsymbol{W} \boldsymbol{g}_{p^n}^H - A_s^2) \geq 0, \nonumber \\
	 &\qquad \qquad \qquad \qquad n \in \{1,2,..,N_\text{e}\}, \label{Eqn:MimoSuboptimalFeasibility_C2a} \\
	 & \Tr{\boldsymbol{W}} \leq \nu, \label{Eqn:MimoSuboptimalFeasibility_C2b}
	\end{align}	
	\label{Eqn:MimoSuboptimalFeasibility}
\end{subequations}
\noindent\hspace*{-3pt}where $\boldsymbol{u} = [u_1, u_2, .., u_{N_\text{e}}]^\top = [\boldsymbol{0}_k^\top \boldsymbol{1}_{N_\text{e}-k}^\top]^\top$ in (\ref{Eqn:MimoSuboptimalFeasibility_C2a}) determines the sets of rectifiers, which are and are not driven into saturation, respectively.
In the following theorem, we show that if (\ref{Eqn:MimoSuboptimalFeasibility}) is feasible and $k > 0$, a beamforming matrix $\boldsymbol{W}_{\!\!k}^*$ obtained as solution of (\ref{Eqn:MimoSuboptimalFeasibility}) has rank one.
\begin{proposition}
	If (\ref{Eqn:MimoSuboptimalFeasibility}) is feasible and $k > 0$, the optimal beamforming matrix $\boldsymbol{W}_{\!\!k}^*$ as solution of (\ref{Eqn:MimoSuboptimalFeasibility}) satisfies $\rank \boldsymbol{W}_{\!\!k}^* = 1$.
\end{proposition}
\begin{proof}
	Please refer to Appendix~\ref{Appendix:TheoremProof}.
\end{proof}
Next, we denote by $k^*$ the maximum number of rectifiers $k$, for which problem (\ref{Eqn:MimoSuboptimalFeasibility}) is feasible.
Note that if (\ref{Eqn:MimoSuboptimalFeasibility}) is not feasible for any $k > 0$, we have $k^* = 0$.
Then, we define the convex subset $\mathcal{W}^*$ that corresponds to the case, where the $k^*$ rectifiers with the best channel conditions are driven into saturation. 
This set is given by $\mathcal{W}^* = \big\{\boldsymbol{W} \, | \, (-1)^{u^*_n} (\boldsymbol{g}_{p^n} \boldsymbol{W} \boldsymbol{g}_{p^n}^H - A_s^2) \geq 0, \boldsymbol{u}^* = [\boldsymbol{0}_{k^*}^\top \boldsymbol{1}_{N_\text{e}-{k^*}}^\top]^\top, n \in \{1,2,..,N_\text{e}\}, \boldsymbol{W}\noindent\in\noindent\mathcal{S}_{+} \big\}$, where $\boldsymbol{u}^* = [\boldsymbol{0}_{k^*}^\top \boldsymbol{1}_{N_\text{e}-k^*}^\top]^\top$.
Next, we reformulate problem (\ref{Eqn:MimoSuboptimalFunctionRef}) as follows:
\begin{equation}
	\maximize_{\boldsymbol{W} \in \mathcal{W}^*} \;  \Psi(\boldsymbol{W}) \quad
	\subjectto \; (\text{\upshape\ref{Eqn:MimoSuboptimalFeasibility_C2b}}).
	\label{Eqn:MimoSuboptimalConvexRef}
\end{equation}
Optimization problem (\ref{Eqn:MimoSuboptimalConvexRef}) is still non-convex due to the non-concavity of the objective function. 
Therefore, in the following, we propose to solve (\ref{Eqn:MimoSuboptimalConvexRef}) exploiting \gls*{sca} \cite{Sun2017}.
To this end, we construct an underestimate of the objective function $\Psi(\boldsymbol{W})$, which is convex in $\mathcal{W}^*$, as follows:
\begin{equation}
{\Psi}(\boldsymbol{W}) \geq \hat{\Psi}(\boldsymbol{W}, \boldsymbol{W}^{(t)}),
\end{equation}
\noindent where $\boldsymbol{W}^{(t)}$ is the solution obtained in iteration $t$ of the algorithm, $\hat{\Psi}(\boldsymbol{W}, \boldsymbol{W}^{(t)}) = \Psi(\boldsymbol{W}^{(t)}) + \Tr{ \triangledown{\Psi}(\boldsymbol{W}^{(t)})^H (\boldsymbol{W} - \boldsymbol{W}^{(t)} ) }$, and $\triangledown{\Psi}(\boldsymbol{W}^{(t)})$ is the gradient of $\Psi(\cdot)$ evaluated at $\boldsymbol{W}^{(t)}$ \cite{Ghanem2020a}.
Thus, in each iteration $t$ of the proposed algorithm, we solve the following optimization problem:
\begin{equation}
\begin{aligned}
\boldsymbol{W}^{(t+1)} = \argmax_{\boldsymbol{W} \in {\mathcal{W}^*}  } \, \hat{\Psi}(\boldsymbol{W}, \boldsymbol{W}^{(t)})  \;
\subjectto \;  (\text{\upshape\ref{Eqn:MimoSuboptimalFeasibility_C2b}}).	
\end{aligned}	
\label{Eqn:MimoSuboptimalFunctionAlg}
\end{equation}
\noindent We note that (\ref{Eqn:MimoSuboptimalFunctionAlg}) is a feasible convex optimization problem and can be solved with standard numerical optimization tools, such as CVX \cite{Grant2015}.
Furthermore, it can be shown that the solution of (\ref{Eqn:MimoSuboptimalFunctionAlg}) yields a matrix, whose rank is equal to one.
The corresponding proof is similar to the one in Appendix~\ref{Appendix:TheoremProof}.
Hence, we obtain the beamforming vector $\boldsymbol{w}^*$ as the dominant eigenvector of $\boldsymbol{W}^*$ and compute the corresponding value of function $\Phi(\nu)~=~\psi(\boldsymbol{w}^*)$.
The proposed algorithm is summarized in Algorithm~\ref{OptimizationAlgorithmFunction}.
We note that the proposed algorithm converges to a stationary point of (\ref{Eqn:MimoSuboptimalFunctionRef}) \cite{Lanckriet2009}.
The computational complexity of a single iteration of the algorithm is given by\footnotemark $\mathcal{O}(N_\text{e} N_\text{t}^{\frac{7}{2}} + N_\text{e}^2 N_\text{t}^{\frac{5}{2}} +\sqrt{N_\text{t}} N_\text{e}^3)$, where $\mathcal{O}(\cdot)$ is the big-O notation.
\footnotetext{The computational complexity of a convex semidefinite problem that involves an $n \times n$ positive semidefinite matrix and $m$ constraints is given by $\mathcal{O}\big(\sqrt{n} (m n^3 + m^2 n^2 + m^3)\big)$ \cite{Polik2010}. Here, $n = N_\text{t}$ and $m = N_\text{e}+1$.}

\begin{algorithm}[!t]				
	\SetAlgoNoLine%
	\SetKwFor{Foreach}{for each}{do}{end}		
	Initialize: Transmit power $\nu$, tolerance error $\epsilon_\text{SCA}$.	\\	
	1. Sort the channel gain vectors by their norms $\norm{\boldsymbol{g}_{p^1}} \geq \norm{\boldsymbol{g}_{p^2}} \geq ... \geq \norm{\boldsymbol{g}_{p^{N_\text{e}}}}$, where $p^k \in \{1,2,...,N_\text{e}\}$ and $k = 1,2,...,N_\text{e}$.\\
	2. Set iteration index $j = 1$ and initial value $k^* = 0$\\
	\For{$j = 1$ {\upshape to} $N_\text{e}$}{		
		3. Solve optimization problem (\ref{Eqn:MimoSuboptimalFeasibility}) for $k = j$ and store $k^* = j$ if the problem is feasible \\				
	}
	4.	Determine $\mathcal{W}^*$, set initial values $h^{(0)} = 0$, $t = 0$, and randomly initialize $\boldsymbol{W}^{(0)}$  \\
		\Repeat{$|h^{(t)}-h^{(t-1)}|\leq \epsilon_\text{\upshape{SCA}}$ }{		
			a. For given $\boldsymbol{W}^{(t)}$, obtain $\boldsymbol{W}^{(t+1)}$ as solution of (\ref{Eqn:MimoSuboptimalFunctionAlg}) \\								
			b. Evaluate $h^{(t+1)} = {\Psi}(\boldsymbol{W}^{(t+1)})$\\
			c. Set $t = t+1$\\ 
		}	
	5. Obtain $\boldsymbol{w}^*$ as the dominant eigenvector of $\boldsymbol{W}^{(t)}$ and evaluate $\Phi(\nu) = \psi(\boldsymbol{w}^*)$\\
	\textbf{Output:} 
	$ \Phi(\nu)$, $\boldsymbol{w}^*$
	\caption{\strut Suboptimal algorithm for solving optimization problem (\ref{Eqn:MimoPropositionFunction}) }
	\label{OptimizationAlgorithmFunction}
\end{algorithm}	
			\subsubsection{Grid Search Method}				
			\label{Section:GridSearch}
			In the following, we propose a grid-search based method for solving the min-max optimization problem in Lemma~\ref{Theorem:Corollary2}.
We note that this problem is not convex since function $s(\nu_1, \nu_2)$ is neither convex nor concave in $\nu_1$ and $\nu_2$, respectively. 
However, since the dimensionality of the problem is low, performing a grid search to determine $\nu_1^*$ and $\nu_2^*$ has a limited affordable complexity \cite{Coope2001}.
To this end, we define a uniform grid $\mathcal{P} = \{\rho_0, \rho_1, \rho_2, ..., \rho_{N_\rho}\}$, where $\rho_0 = 0$, $\rho_j = \Delta_\rho + \rho_{j-1}$, $j = 1,2,...,{N_\rho}$, and $\Delta_\rho$ is a predefined step size.
$N_\rho$ is the grid size that is chosen sufficiently large in our simulations to ensure that the function $\psi(\cdot)$ saturates, i.e., all the rectifiers are driven into saturation for $\nu = \rho_{N_\rho}$, cf. (\ref{Eqn:RaniaModel}).
Then, we define the smallest element of $\mathcal{P}$, which is larger than $P_x$ as $\rho_n$, i.e., $\rho_n = \min\{\rho_j | \rho_j \geq P_x, j = 0,1,...,N_\rho\}$.
Next, we define a matrix $\boldsymbol{S} \in \mathbb{R}^{n\times (N_\rho-n+1)}$, whose elements are the values of function $s(\cdot, \cdot)$ evaluated at the elements of $\mathcal{P}$, i.e., ${S}_{i,j} = s(\rho_i, \rho_{j'})$, $i = 0,1,...,{n-1}$, $j = {j'}-n$, and ${j'} = n,n+1,...,N_\rho$.
Finally, we obtain the power values $\nu_1^* = \rho_{i^*}$ and $\nu_2^*=\rho_{n+j^*}$, where $i^* = \argmin_{i} \max_{j} {S}_{i,j}$ and $j^* = \argmax_{j} {S}_{i^*,j}$, and the corresponding beamforming vectors $\boldsymbol{w}_1^*$ and $\boldsymbol{w}_2^*$.
The proposed scheme is summarized in Algorithm~\ref{AlgorithmGridSearch}.
The computational complexity of the proposed scheme is quadratic with respect to the grid size $N_\rho$ and does not depend on the numbers of antennas $N_\text{t}$ and $N_\text{e}$.
\begin{algorithm}[!t]				
	\SetAlgoNoLine%
	\SetKwFor{Foreach}{for each}{do}{end}		
	Initialize: Grid size $N_{\rho}$, step size $\Delta_{\rho}$, maximum TX power $P_x$, initial value $\rho_0=0$.	\\	
	1. Compute the grid $\mathcal{P}$ and the values of $\Phi(\cdot)$ for the grid elements:\\
	\For{$m = 0$ {\normalfont to} $N_{\rho}$}{
		1.1. Compute $\Phi_{m}$ = $\Phi(\rho_{m})$ and $\boldsymbol{v}_m = \boldsymbol{w}^*$ with Algorithm~\ref{OptimizationAlgorithmFunction}\\
		1.2. Set $\rho_{m+1} = \rho_{m} + \Delta_{\rho}$ \\
	}
	2. Determine the grid element $\rho_n = \min\{\rho_j | \rho_j \geq P_x, j = 0,1,...,N_\rho \}$ \\
	3. Calculate the elements of matrix $\boldsymbol{S}$ as ${S}_{i,j} = s(\rho_i, \rho_{j'}) = \frac{\Phi_{j'} - \Phi_i}{\rho_{j'} - \rho_i}$, $i = 0,1,...,{n-1}$, $j = {j'}-n$, and ${j'} = n,n+1,...,N_\rho$ \\
	4. Determine the power values and the corresponding vectors $\nu_1^* = \rho_{i^*}, \boldsymbol{w}_1^* = \boldsymbol{v}_{i^*}$ and $\nu_2^*=\rho_{n+j^*}, \boldsymbol{w}_2^* = \boldsymbol{v}_{n+j^*}$, where $i^* = \argmin_{i} \max_{j} {S}_{i,j}$ and $j^* = \argmax_{j} {S}_{i^*,j}$ \\
	\textbf{Output:} $p^*_{\boldsymbol{w}}(\boldsymbol{w}) = \beta \delta(\boldsymbol{w} - \boldsymbol{w}^*_1) + (1-\beta) \delta(\boldsymbol{w} - \boldsymbol{w}^*_2)$ with $\beta = \frac{\nu_2^* - P_x}{\nu_2^* - \nu_1^*}$
	\caption{\strut Grid search for determining the pdf $p^*_{\boldsymbol{w}}(\boldsymbol{w})$ }
	\label{AlgorithmGridSearch}
\end{algorithm}	 
				
	\section{Numerical Results}
	In this section, we evaluate the performance of the proposed transmit strategies via simulations.
In our simulations, the path losses are calculated as $35.3+37.6\log_{10}(d)$, where $d = \SI{10}{\meter}$ is the distance between the TX and the EH node \cite{Ghanem2020a}.
Furthermore, in order to harvest a meaningful amount of power, we assume that the TX and EH node have a line-of-sight link.
Thus, the channel gains $\boldsymbol{g}_p$ follow Rician distributions with Rician factor $1$.
For the EH model $\phi(\cdot)$ in (\ref{Eqn:RaniaModel}), we adopt the following values of the parameters {$a=1.29$, $B = 1.55\cdot10^3$, $I_s = \SI{5}{\micro\ampere}$, $R_L = \SI{10}{\kilo\ohm}$, and $A_s^2 = \SI{25}{\micro\watt} $}, respectively.
For Algorithms~\ref{OptimizationAlgorithmFunction} and \ref{AlgorithmGridSearch}, we adopt step size $\Delta_{\rho} = 0.1$, grid size $N_{\rho} = 10^3$, and tolerance error $\epsilon_{\text{SCA}} = 10^{-3}$.
All simulation results were averaged over $100$ channel realizations.

In Fig.~\ref{Fig:Results_SU}, we investigate the performance of various WPT systems for different values of the power budget $P_x$.
The considered \gls*{miso} and \gls*{simo} WPT systems employ $N_\text{t} = 2$ and $N_\text{e} = 2$ antennas at the TX and the EH, respectively, whereas for the MIMO WPT system, both the TX and EH have $N_\text{t} = N_\text{e} = 2$ antennas.
For these systems, exploiting Algorithms~\ref{OptimizationAlgorithmFunction} and \ref{AlgorithmGridSearch}, we determine the pdfs $p^*_{\boldsymbol{x}}(\boldsymbol{x})$ and compute the corresponding average harvested powers $\overline{\Phi}(\boldsymbol{x}; p^*_{\boldsymbol{x}})$.
For comparison, we also consider a SISO WPT system employing the optimal transmit strategy in \cite{Morsi2019}.
As Baseline Scheme 1, we consider a MIMO WPT system with energy beamforming at the TX, which is optimal for linear EHs \cite{Zhang2013}.
Furthermore, as Baseline Scheme 2, similar to \cite{Shen2020}, we consider a MIMO WPT system, where a scalar input symbol and a single beamforming vector are adopted at the TX. 
For this system, we obtain the optimal beamforming vector $\boldsymbol{w}^*$ as solution of (\ref{Eqn:MimoPropositionBeamformerProblem}) for $\nu = P_x$ with Algorithm~\ref{OptimizationAlgorithmFunction} and compute the corresponding harvested power as $\overline{\Phi} = \Phi(\boldsymbol{w}^*)$.

In Fig.~\ref{Fig:Results_SU}, we observe that for each considered WPT setup, the average harvested power $\overline{\Phi}(\cdot)$ is bounded above for sufficiently large values of $P_x$, i.e., all rectifiers of the EH node are driven into saturation, cf. (\ref{Eqn:RaniaModel}).
Furthermore, the saturation level of the harvested power is proportional to the number of rectennas employed at the EH.
As expected, the MIMO WPT system achieves a superior performance compared to the SIMO and MISO WPT systems, which, in turn, outperform the SISO WPT system  significantly.
Interestingly, the MISO WPT system can harvest more power than the SIMO WPT system if the transmit power budget at the TX is low, whereas, for large values of $P_x$, the single rectenna of the MISO WPT system is driven into saturation and, thus, more power can be harvested by the SIMO WPT system.
Finally, we observe that for the MIMO WPT system, the proposed optimal transmit strategy, which employs two beamforming vectors, outperforms Baseline Schemes 1 and 2, which employ a single beamforming vector.
However, we note that Baseline Scheme 2, where the TX is equipped with the optimal beamforming vector outperforms Baseline Scheme 1, where energy beamforming is adopted \cite{Zhang2013}.
\begin{figure}[!t]
	\centering
	\includegraphics[width=0.45\textwidth, draft=false]{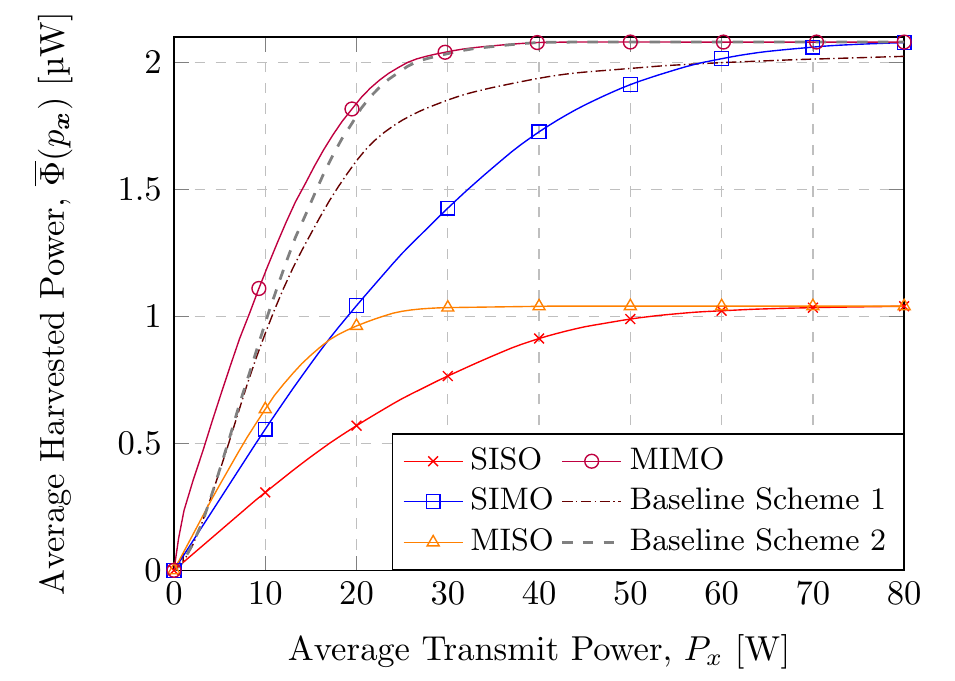}
	\caption{Comparison of SISO, SIMO, MISO, and MIMO WPT systems.}
	\label{Fig:Results_SU}
	\vspace*{-10pt}
\end{figure}
\begin{figure*}[!t]
	\centering
	\subfigure[Comparison for different numbers of EH antennas, $N_{\text{e}}$
	]{
		\includegraphics[width=0.45\textwidth, draft=false]{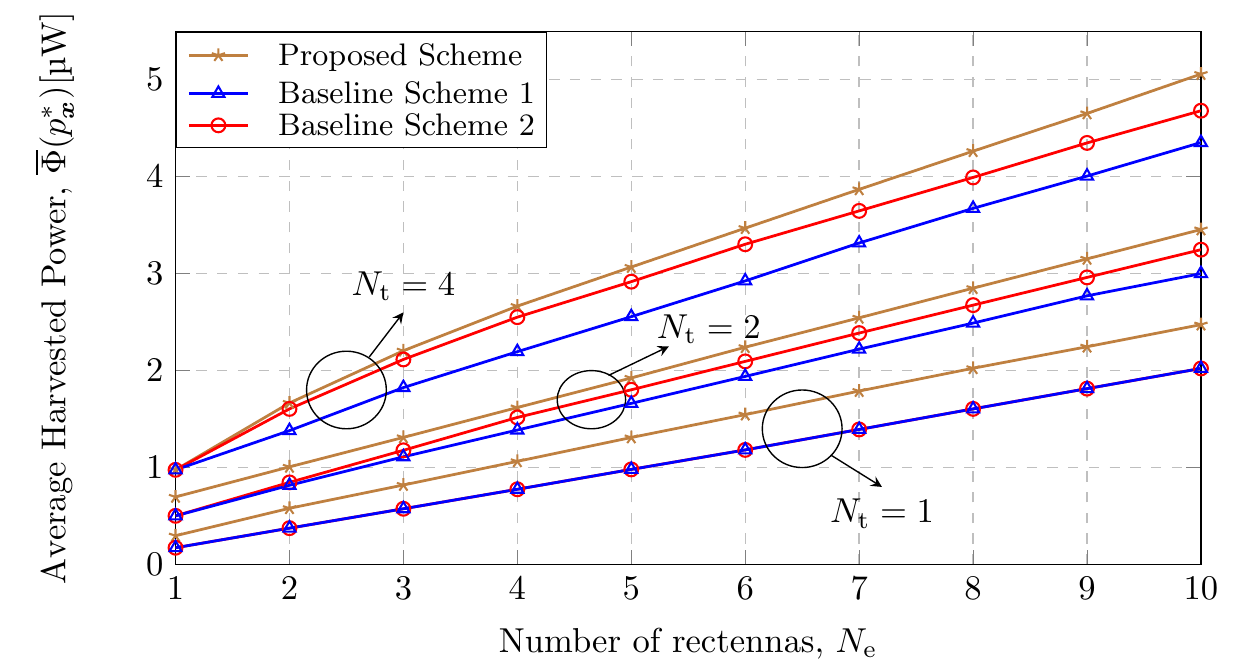} \label{Fig:Results_MU_Neh}}
	\quad
	\subfigure[Comparison for different numbers of TX antennas, $N_{\text{t}}$
	]{
		\includegraphics[width=0.45\textwidth, draft=false]{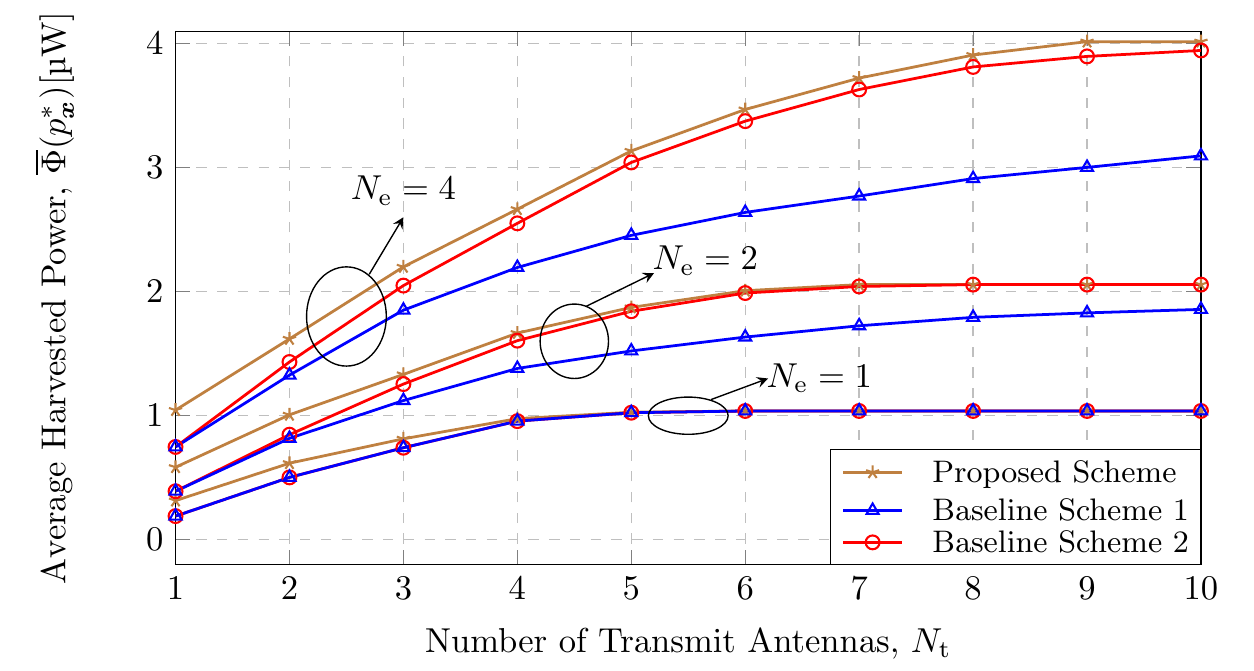} \label{Fig:Results_MU_Nt}}
	\caption{Average harvested power ${\overline{\Phi}(\boldsymbol{x}; p^*_{\boldsymbol{x}})}$ for different numbers of antennas $N_\text{e}$ and $N_\text{t}$.}
	\label{Fig:Results_MU}
	\vspace*{-10pt}
\end{figure*}

In Fig.~\ref{Fig:Results_MU_Neh} and \ref{Fig:Results_MU_Nt}, we compare the average harvested powers for different numbers of antennas $N_\text{e}$ and $N_\text{t}$, respectively, assuming a transmit power budget of $P_x = \SI{10}{\watt}$.
For each system setup, we compare the performance of the proposed transmit strategy with Baseline Schemes 1 and 2.
First, we observe that higher values of $N_\text{e}$ and $N_\text{t}$ yield larger average harvested powers $\overline{\Phi}$.
Moreover, we note that for SIMO WPT systems, i.e., for $N_\text{t} = 1$, the transmit strategies for Baseline Scheme 1 and Baseline Scheme 2 are identical and depend only on the power budget $P_x$.
In Fig.~\ref{Fig:Results_MU_Neh}, we observe that the harvested power depends practically linearly on the number of antennas at the EH node even if the rectifiers are not driven into saturation.
However, in Fig.~\ref{Fig:Results_MU_Nt}, we observe that for large numbers of transmit antennas, the average harvested power saturates since the harvested power in (\ref{Eqn:RaniaModel}) is bounded above. 
Furthermore, we observe that the number of transmit antennas, which is required to drive the rectennas into saturation, grows with $N_\text{e}$.
In fact, larger numbers of antennas $N_\text{t}$ enable a more efficient allocation of the power available at the TX, which, in turn, yields higher received powers at the EH node.

	\section{Conclusion}
	In this paper, we considered MIMO WPT systems with multiple non-linear rectennas at the EH node. 
We formulated an optimization problem for the maximization of the average harvested power subject to the power budget at the TX.
We showed that the optimal transmit strategy employs scalar unit-norm input symbols with an arbitrary phase and two beamforming vectors.
In order to obtain these vectors, we proposed a non-convex optimization problem and a corresponding iterative algorithm based on a two-dimensional grid search, \gls*{sdr}, and \gls*{sca} to solve it.
Our simulation results revealed that the proposed MIMO WPT design outperforms two baseline schemes based on a single beamforming vector and a linear EH model, respectively.
Furthermore, we observed that the average harvested power grows linearly with the number of antennas at the EH and saturates for a large number of antennas equipped at the TX.
			
	\appendices
	\begin{appendices}
		\renewcommand{\thesection}{\Alph{section}}
		\renewcommand{\thesubsection}{\thesection.\arabic{subsection}}
		\renewcommand{\thesectiondis}[2]{\Alph{section}:}
		\renewcommand{\thesubsectiondis}{\thesection.\arabic{subsection}:}	
		\section{Proof of Proposition 1}
		\label{Appendix:Prop3Proof}
			First, we note that for any arbitrary transmit symbol $\tilde{\boldsymbol{x}}$, there is a symbol $\hat{\boldsymbol{x}}$ given by
	\begin{equation}
		\hat{\boldsymbol{x}} = \argmax_{\boldsymbol{x}} \psi(\boldsymbol{x}) \; \subjectto \norm{\boldsymbol{x}}^2 = \norm{\tilde{\boldsymbol{x}}}^2,
	\end{equation}  
	which has the same transmit power and yields a higher or equal value of $\psi(\boldsymbol{x})$.
	Hence, for any arbitrary distribution of transmit symbols with a point of increase at $\tilde{\boldsymbol{x}}$, a larger value of $\Psi(\boldsymbol{x})$ can be obtained by removing this point and increasing the probability of $\hat{\boldsymbol{x}}$ by the corresponding value.
	
	Let us introduce now the monotonically non-decreasing function $\Phi(\nu)$ in (\ref{Eqn:MimoPropositionFunction}) that returns the largest possible value of $\psi(\boldsymbol{x})$ if a symbol with power $\nu$ was transmitted. 
	Then, the solution of (\ref{Eqn:WPT_GeneralProblem}) can be obtained by determining first the solution $p^*_\nu(\nu)$ of the following optimization problem:
	\begin{equation}
		\maximize_{p_{\nu}} \; \mathbb{E}_\nu\{\Phi(\nu)\}\; \quad \subjectto \; \mathbb{E}_\nu\{\nu\} \leq P_x.
		\label{Eqn:MimoPropositionProofProblem}
	\end{equation}
	Since (\ref{Eqn:MimoPropositionProofProblem}) is in the form of (\ref{Eqn:GeneralOptimizationProblem}), $p^*_\nu(\nu)$ is a discrete pdf consisting of two mass points, $\nu^*_1$ and $\nu^*_2$, see Lemma~\ref{Theorem:Corollary2}.
	Hence, the optimal symbol vectors $\boldsymbol{x}$ can be decomposed as $\boldsymbol{x} = \boldsymbol{w} s$ with unit-norm symbols $s$, where the pdf of discrete random beamforming vector $\boldsymbol{w}$ also consists of two mass points evaluated as 
	\begin{equation}
	 {\boldsymbol{w}^*_n} = \argmax_{\boldsymbol{w}: \, \norm{\boldsymbol{w}}^2 = \nu^*_n} \psi(\boldsymbol{w}), n\in\{1,2\},
	\end{equation}	
	with probabilities $p^*_{\boldsymbol{w}}(\boldsymbol{w}^*_n) = p^*_{\nu}(\nu^*_n), n\in\{1,2\}$, respectively.
	Finally, applying Lemma~\ref{Theorem:Corollary2} to optimization problem (\ref{Eqn:MimoPropositionProofProblem}), yields Proposition~\ref{Theorem:Proposition3}.
	This concludes the proof.
		
		\section{Proof of Proposition 2}
		\label{Appendix:TheoremProof}
		First, we rewrite convex problem (\ref{Eqn:MimoSuboptimalFeasibility}) equivalently as follows:
\begin{subequations}
	\begin{align}
	\maximize_{\boldsymbol{W}} \quad &1\\
	\subjectto \quad & \boldsymbol{g}_m \boldsymbol{W} \boldsymbol{g}_m^H \leq A_s^2, \; \forall m \in \mathcal{M} \label{Eqn:RankProofProblemLambdaConstr}\\
	& \boldsymbol{g}_n \boldsymbol{W} \boldsymbol{g}_n^H \geq A_s^2, \; \forall n \in \mathcal{N} \label{Eqn:RankProofProblemMuConstr}\\
	& \Tr{\boldsymbol{W}} \leq \nu \label{Eqn:RankProofProblemEpsConstr}\\
	& \boldsymbol{W} \in \mathcal{S}_{+} \label{Eqn:RankProofProblemYConstr},
	\end{align}
	\label{Eqn:RankProofProblem}
\end{subequations}
\noindent\hspace*{-4pt}where $\mathcal{N} = \{p^1, p^2,.., p^k\}$ and $\mathcal{M} = \{p^{k+1}, p^{k+2}, .., p^{N_\text{e}}\}$ contain the indices of the rectifiers that are and are not driven into saturation, respectively.
Since (\ref{Eqn:RankProofProblem}) is feasible and $k > 0$, set $\mathcal{N}$ is not empty and $\nu > 0$.
Furthermore, the strong duality holds and the gap between (\ref{Eqn:RankProofProblem}) and its dual problem is equal to zero \cite{Boyd2004}.
The Lagrangian of (\ref{Eqn:RankProofProblem}) is given by:
\begin{align}
	\mathcal{L} = & \sum_m \lambda_m \boldsymbol{g}_m \boldsymbol{W} \boldsymbol{g}_m^H - \sum_n \mu_n \boldsymbol{g}_n \boldsymbol{W} \boldsymbol{g}_n^H \nonumber \\ & + \xi \Tr{\boldsymbol{W}} - \Tr{\boldsymbol{Y}\boldsymbol{W}} + {\gamma},
	\label{Eqn:RankProofLagrangian}
\end{align}
\noindent where ${\lambda_m}, m\in\mathcal{M}$, ${\mu_n},  n\in\mathcal{N}$, $\xi$, and $\boldsymbol{Y}$ are the Lagrangian multipliers associated with constraints (\ref{Eqn:RankProofProblemLambdaConstr}), (\ref{Eqn:RankProofProblemMuConstr}), (\ref{Eqn:RankProofProblemEpsConstr}), and (\ref{Eqn:RankProofProblemYConstr}), respectively, and $\gamma$ accounts for all terms that do not involve $\boldsymbol{W}$.
Next, we note that the Karush-Kuhn-Tucker (KKT) conditions are satisfied for the optimal solution of (\ref{Eqn:RankProofProblem}) denoted by $\boldsymbol{W}_{\!\!k}^*$ and the solutions  $\boldsymbol{\lambda}^*$, $\boldsymbol{\mu}^*$, $\xi^*$, and $\boldsymbol{Y}^*$ of the corresponding dual problem.
These conditions are given by
\begin{subequations}
	\begin{align}
	&\triangledown \mathcal{L} = 0 
	\label{Eqn:RankProofKKTGradient}\\
	& \lambda_m^* \geq 0, \mu_n^* \geq 0, \xi^* \geq 0, \boldsymbol{Y}^* \succeq 0, \; \forall m \in \mathcal{M}, \, \forall n \in \mathcal{N} 
	\label{Eqn:RankProofKKTPositive}\\
	& \boldsymbol{Y}^* \boldsymbol{W}_{\!\!k}^* = \boldsymbol{0}_{N_\text{t} \times N_\text{t}},
	\label{Eqn:RankProofKKTSemidef}
	\end{align}
	\label{Eqn:RankProofKKT}
\end{subequations}
\noindent\hspace*{-4pt}where $\boldsymbol{0}_{N \times N}$ denotes the square all-zero matrix of size $N$.
Next, we express condition (\ref{Eqn:RankProofKKTGradient}) as follows:
\begin{equation}
	\boldsymbol{Y}^* = \xi^* \boldsymbol{I} - \boldsymbol{\Delta}, \label{Eqn:RankProofMainEquality} 
\end{equation}
\noindent where $\boldsymbol{\Delta} = \sum_n \mu_n^* \boldsymbol{g}_n^H \boldsymbol{g}_n - \sum_m \lambda_m^* \boldsymbol{g}_m^H \boldsymbol{g}_m$ and $\boldsymbol{I}$ denotes the identity matrix.
Let us now investigate the structure of $\boldsymbol{\Delta}$.
We denote the maximum eigenvalue of $\boldsymbol{\Delta}$ by $\delta^\text{max} \in \mathbb{R}$.
Due to the randomness of the channel, with probability $1$, only one eigenvalue of $\boldsymbol{\Delta}$ has value $\delta^\text{max}$.
Observing (\ref{Eqn:RankProofMainEquality}), we note that if $\xi^* > \delta^\text{max}$, then $\boldsymbol{Y}^*$ is a positive semidefinite matrix with full rank.
In this case, (\ref{Eqn:RankProofKKTSemidef}) yields $\boldsymbol{W}_{\!\!k}^* = \boldsymbol{0}_{N_\text{t} \times N_\text{t}}$ and, hence, $\rank \boldsymbol{W}_{\!\!k}^* = 0$, which is a feasible solution of (\ref{Eqn:RankProofProblem}) if and only if $k = 0$.
Furthermore, if $\xi^* < \delta^\text{max}$, then $\boldsymbol{Y}^*$ is not a positive semidefinite matrix, which contradicts (\ref{Eqn:RankProofKKTPositive}).
Hence, for $k > 0$, $\xi^* = \delta^\text{max} \geq 0$ and $\boldsymbol{Y}^*$ is a positive semidefinite matrix with $\rank \{\boldsymbol{Y}^*\} = N_\text{t} - 1$.
Then, applying Sylvester's rank inequality to (\ref{Eqn:RankProofKKTSemidef}), we have
\begin{align}
	0 = \rank \{\boldsymbol{Y}^* \boldsymbol{W}_{\!\!k}^*\} \geq \rank \{\boldsymbol{Y}^*\} &+ \rank \{\boldsymbol{W}_{\!\!k}^*\} - N_\text{t} \nonumber \\ &= \rank \{\boldsymbol{W}_{\!\!k}^*\} - 1.
\end{align}
Finally, since there exists a feasible solution of (\ref{Eqn:MimoSuboptimalFeasibility}) and $k > 0$, $\rank \{\boldsymbol{W}_{\!\!k}^*\} = 1$. This concludes the proof.

	\end{appendices}

	\bibliographystyle{IEEEtran}
	\bibliography{Final}

\begin{thebibliography}{10}
\providecommand{\url}[1]{#1}
\csname url@samestyle\endcsname
\providecommand{\newblock}{\relax}
\providecommand{\bibinfo}[2]{#2}
\providecommand{\BIBentrySTDinterwordspacing}{\spaceskip=0pt\relax}
\providecommand{\BIBentryALTinterwordstretchfactor}{4}
\providecommand{\BIBentryALTinterwordspacing}{\spaceskip=\fontdimen2\font plus
\BIBentryALTinterwordstretchfactor\fontdimen3\font minus
  \fontdimen4\font\relax}
\providecommand{\BIBforeignlanguage}[2]{{%
\expandafter\ifx\csname l@#1\endcsname\relax
\typeout{** WARNING: IEEEtran.bst: No hyphenation pattern has been}%
\typeout{** loaded for the language `#1'. Using the pattern for}%
\typeout{** the default language instead.}%
\else
\language=\csname l@#1\endcsname
\fi
#2}}
\providecommand{\BIBdecl}{\relax}
\BIBdecl

\bibitem{Grover2010}
P.~{Grover} and A.~{Sahai}, ``{Shannon} meets {Tesla}: {Wireless} information
  and power transfer,'' in \emph{Proc. IEEE Int. Symp. Information Theory},
  Jun. 2010, pp. 2363--2367.

\bibitem{Zhang2013}
R.~{Zhang} and C.~K. {Ho}, ``{MIMO} broadcasting for simultaneous wireless
  information and power transfer,'' \emph{IEEE Trans. Wirel. Commun.}, vol.~12,
  no.~5, pp. 1989--2001, May 2013.

\bibitem{Boshkovska2015}
E.~{Boshkovska}, D.~W.~K. {Ng}, N.~{Zlatanov}, and R.~{Schober}, ``Practical
  non-linear energy harvesting model and resource allocation for {SWIPT}
  systems,'' \emph{IEEE Commun. Lett.}, vol.~19, no.~12, pp. 2082--2085, Dec.
  2015.

\bibitem{Kim2020}
J.~Kim, B.~Clerckx, and P.~D. Mitcheson, ``Signal and system design for
  wireless power transfer: Prototype, experiment and validation,'' \emph{{IEEE}
  Trans. Wirel. Commun.}, vol.~19, no.~11, pp. 7453--7469, Nov. 2020.

\bibitem{Boshkovska2017a}
E.~Boshkovska, D.~W.~K. Ng, N.~Zlatanov, A.~Koelpin, and R.~Schober, ``Robust
  resource allocation for {MIMO} wireless powered communication networks based
  on a non-linear {EH} model,'' \emph{{IEEE} Trans. Commun.}, vol.~65, no.~5,
  pp. 1984--1999, May 2017.

\bibitem{Ma2019}
G.~Ma, J.~Xu, Y.~Zeng, and M.~R.~V. Moghadam, ``A generic receiver architecture
  for {MIMO} wireless power transfer with nonlinear energy harvesting,''
  \emph{{IEEE} Signal Process. Lett.}, vol.~26, no.~2, pp. 312--316, Feb. 2019.

\bibitem{Clerckx2018}
B.~{Clerckx}, ``Wireless information and power transfer: Nonlinearity, waveform
  design, and rate-energy tradeoff,'' \emph{IEEE Trans. Signal Process.},
  vol.~66, no.~4, pp. 847--862, Feb. 2018.

\bibitem{Shen2020}
S.~Shen and B.~Clerckx, ``Beamforming optimization for {MIMO} wireless power
  transfer with nonlinear energy harvesting: {RF} combining versus {DC}
  combining,'' \emph{IEEE Trans. Wirel. Commun.}, vol.~20, no.~1, pp. 199--213,
  Jan. 2021.

\bibitem{Morsi2019}
R.~{Morsi}, V.~{Jamali}, A.~{Hagelauer}, D.~W.~K. {Ng}, and R.~{Schober},
  ``Conditional capacity and transmit signal design for {SWIPT} systems with
  multiple nonlinear energy harvesting receivers,'' \emph{IEEE Trans. Commun.},
  vol.~68, no.~1, pp. 582--601, Jan. 2020.

\bibitem{Shanin2020}
N.~Shanin, L.~Cottatellucci, and R.~Schober, ``Markov decision process based
  design of {SWIPT} systems: Non-linear {EH} circuits, memory, and impedance
  mismatch,'' \emph{{IEEE} Trans. Commun.}, vol.~69, no.~2, pp. 1259 -- 1274,
  Feb. 2021.

\bibitem{Shanin2021a}
------, ``Harvested power region of two-user {MISO WPT} systems with non-linear
  {EH} nodes,'' \emph{arXiv preprint arXiv:2103.13802}, 2021.

\bibitem{Grippo2000}
L.~Grippo and M.~Sciandrone, ``On the convergence of the block nonlinear
  {Gauss-Seidel} method under convex constraints,'' \emph{Operations Research
  Letters}, vol.~26, no.~3, pp. 127--136, 2000.

\bibitem{Xu2019}
D.~Xu, X.~Yu, Y.~Sun, D.~W.~K. Ng, and R.~Schober, ``Resource allocation for
  secure {IRS}-assisted multiuser {MISO} systems,'' in \emph{{IEEE} Globecom
  Workshop}, Dec. 2019.

\bibitem{Sun2017}
Y.~Sun, P.~Babu, and D.~P. Palomar, ``Majorization-minimization algorithms in
  signal processing, communications, and machine learning,'' \emph{{IEEE}
  Trans. Signal Process.}, vol.~65, no.~3, pp. 794--816, Feb. 2017.

\bibitem{Dokov2002}
S.~P. Dokov and D.~P. Morton, \emph{Higher-Order Upper Bounds on the
  Expectation of a Convex Function}.\hskip 1em plus 0.5em minus 0.4em\relax
  Humboldt-Universit{\"a}t zu Berlin, 2002.

\bibitem{Coope2001}
I.~D. Coope and C.~J. Price, ``On the convergence of grid-based methods for
  unconstrained optimization,'' \emph{{SIAM} J. Optim.}, vol.~11, no.~4, pp.
  859--869, Jan. 2001.

\bibitem{Lanckriet2009}
G.~Lanckriet and B.~K. Sriperumbudur, ``On the convergence of the
  concave-convex procedure,'' in \emph{Advances in Neural Information
  Processing Systems}, vol.~22, Dec. 2009, pp. 1759--1767.

\bibitem{Ghanem2020a}
W.~R. Ghanem, V.~Jamali, and R.~Schober, ``Optimal resource allocation for
  multi-user {OFDMA}-{URLLC} {MEC} systems,'' \emph{arXiv preprint
  arXiv:2009.11073}, 2020.

\bibitem{Grant2015}
M.~Grant and S.~Boyd, ``{CVX}: Matlab software for disciplined convex
  programming, version 2.0 beta (2013),'' \emph{URL: http://cvxr. com/cvx},
  2015.

\bibitem{Polik2010}
I.~P{\'o}lik and T.~Terlaky, \emph{Interior Point Methods for Nonlinear
  Optimization}.\hskip 1em plus 0.5em minus 0.4em\relax Springer, 2010.

\bibitem{Boyd2004}
S.~Boyd, S.~P. Boyd, and L.~Vandenberghe, \emph{Convex Optimization}.\hskip 1em
  plus 0.5em minus 0.4em\relax Cambridge University Press, 2004.

\end{thebibliography}

\end{document}